\documentclass{article}
\usepackage{arxiv}

\usepackage[utf8]{inputenc}
\usepackage[T1]{fontenc}
\usepackage{booktabs}
\usepackage{graphicx}
\graphicspath{{code/figures/}}
\usepackage{hyperref}
\usepackage{multirow}
\usepackage{listings}
\usepackage{natbib}
\usepackage{amsmath,amssymb,amsthm}
\usepackage[dvipsnames, table]{xcolor}
\usepackage[english]{babel}
\usepackage{tikz}
\usepackage{subcaption}
\usepackage{xcolor}

\usepackage[framemethod=tikz]{mdframed}
\usepackage{tcolorbox}
\tcbuselibrary{listings, skins, breakable, theorems}

\definecolor{bgdark}{HTML}{FFFFFF}
\definecolor{bgsurface}{HTML}{F6F8FA}
\definecolor{bgsurface2}{HTML}{EAEEF2}
\definecolor{accentblue}{HTML}{0969DA}
\definecolor{accentpurple}{HTML}{8250DF}
\definecolor{accentgreen}{HTML}{1A7F0F}
\definecolor{accentorange}{HTML}{D1521A}
\definecolor{accentred}{HTML}{CF222E}
\definecolor{textmuted}{HTML}{000000}
\definecolor{bordercolor}{HTML}{D0D7DE}
\definecolor{codemono}{HTML}{0969DA}
\definecolor{codestring}{HTML}{0A3069}
\definecolor{codekeyword}{HTML}{CF222E}
\definecolor{codecomment}{HTML}{000000}

\hypersetup{
  colorlinks=true,
  citecolor=codestring,
  linkcolor=codestring,
  urlcolor=codestring,
}

\tcbset{
  defbox/.style={
    enhanced, breakable,
    colback=bgsurface, colframe=accentblue,
    colbacktitle=accentblue!20!bgsurface,
    coltitle=accentblue,
    fonttitle=\bfseries\small,
    left=8pt, right=8pt, top=6pt, bottom=6pt,
    titlerule=0.5pt,
    arc=2pt
  },
  notebox/.style={
    enhanced, breakable,
    colback=accentpurple!8!bgsurface, colframe=accentpurple!60!black,
    left=8pt, right=8pt, top=6pt, bottom=6pt,
    arc=2pt,
    fontupper=\small\color{black}
  },
  warnbox/.style={
    enhanced,
    colback=accentorange!10!bgsurface, colframe=accentorange,
    left=8pt, right=8pt, top=4pt, bottom=4pt,
    arc=2pt,
    fontupper=\small\color{black}
  }
}

\lstdefinelanguage{pseudocode}{
  morekeywords={function, return, if, else, for, in, while, and, or, not,
                class, def, import, from, true, false, None, True, False},
  sensitive=true,
  morecomment=[l]{//},
  morecomment=[l]{\#},
  morestring=[b]",
  morestring=[b]',
}

\lstdefinestyle{darkcode}{
  backgroundcolor=\color{bgsurface},
  basicstyle=\ttfamily\footnotesize\color{black},
  keywordstyle=\color{codekeyword}\bfseries,
  stringstyle=\color{codestring},
  commentstyle=\color{codecomment}\itshape,
  numberstyle=\tiny\color{textmuted},
  breaklines=true,
  breakatwhitespace=false,
  numbers=left,
  numbersep=10pt,
  showspaces=false,
  showstringspaces=false,
  frame=none,
  xleftmargin=16pt,
  xrightmargin=8pt,
  tabsize=2,
  aboveskip=0pt,
  belowskip=0pt,
  literate={->}{{$\to$}}2 {>=}{{$\geq$}}2 {<=}{{$\leq$}}2
           {!=}{{$\neq$}}2 {**}{{${}^{**}$}}2
}

\newtcblisting{codeblock}[2][]{
  enhanced, breakable,
  listing only,
  colback=bgsurface,
  colframe=bordercolor,
  listing options={style=darkcode, language=#2},
  title={\texttt{\textcolor{textmuted}{#1}}},
  fonttitle=\small,
  colbacktitle=bgsurface2,
  coltitle=textmuted,
  arc=2pt,
  top=0pt, bottom=0pt,
  #1
}

\newtheorem{proposition}{Proposition}

\newtheorem{corollary}{Corollary}

\tcolorboxenvironment{definition}{
  enhanced,
  colback=accentblue!10!bgsurface,
  colframe=accentblue,
  boxsep=0pt,
  left=6pt,
  right=6pt,
  top=6pt,
  bottom=6pt,
  arc=2pt,
  before upper={\parindent0pt}
}

\tcolorboxenvironment{proposition}{
  enhanced,
  colback=accentblue!10!bgsurface,
  colframe=accentblue,
  boxsep=0pt,
  left=6pt,
  right=6pt,
  top=6pt,
  bottom=6pt,
  arc=2pt,
  before upper={\parindent0pt}
}

\tcolorboxenvironment{lemma}{
  enhanced,
  colback=accentpurple!10!bgsurface,
  colframe=accentpurple,
  boxsep=0pt,
  left=6pt,
  right=6pt,
  top=6pt,
  bottom=6pt,
  arc=2pt,
  before upper={\parindent0pt}
}

\tcolorboxenvironment{corollary}{
  enhanced,
  colback=accentgreen!10!bgsurface,
  colframe=accentgreen,
  boxsep=0pt,
  left=6pt,
  right=6pt,
  top=6pt,
  bottom=6pt,
  arc=2pt,
  before upper={\parindent0pt}
}

\usepackage{pgfplots}
\pgfplotsset{compat=1.18}
\usepackage{siunitx}
\usepackage{cleveref}   % must be loaded after hyperref

\definecolor{AccentBlue}{HTML}{236CFF}
\definecolor{Midnight}{HTML}{0D2644}
\definecolor{AccentOrange}{HTML}{E87722}
\definecolor{NeutralGray}{HTML}{6B7280}
\definecolor{Cnginx}{HTML}{1A73E8}   % nginx  — blue
\definecolor{Capache}{HTML}{E8710A}  % Apache — orange
\definecolor{Ccaddy}{HTML}{188038}   % Caddy  — green
\definecolor{Cpyv}{HTML}{9AA0A6}     % CPython (all versions) — gray (degenerate)
\definecolor{Cpg}{HTML}{9334E6}      % PostgreSQL — purple
\definecolor{Cmysql}{HTML}{C5221F}   % MySQL — red
\definecolor{Clam}{HTML}{E8710A}     % Serverless runtime bars

\pgfplotsset{
  panel/.style={
    width=\linewidth, height=6.8cm,
    grid=both,
    major grid style={dotted,gray!35},
    minor grid style={dotted,gray!15},
    tick label style={font=\small}, label style={font=\small},
    title style={font=\small\bfseries, color=Midnight},
    legend style={font=\footnotesize, draw=gray!40, fill=white,
                  fill opacity=0.92, inner sep=2pt, row sep=-2pt},
  },
  pool/.style={
    color=Midnight, mark=*, mark size=2.6pt,
    line width=1.4pt, const plot mark right,
  },
}

\title{$\varepsilon$-Indistinguishability in Moving Target Defense: Framework, Algorithms, and Cloud Case Studies}
\author{ 
    Sailik Sengupta \thanks{Equal Contribution}~~\thanks{The work does not relate to the author's position at Amazon.}\\
    Amazon Science\\
    \texttt{sailiks@amazon.com}
    \And
    Ankur Chowdhary$~^*$ \\
    Intuit \\
	\texttt{ankur\_chowdhary@intuit.com}
}
\date{}

\renewcommand{\shorttitle}{$\varepsilon$-Indistinguishability in Moving Target Defense}

\begin{document}
\maketitle

\begin{abstract}
Moving Target Defense (MTD) assumes its pool of candidate configurations is safe to cycle among, i.e. latency and other observables do not trivially fingerprint the active choice, but this assumption has not been quantified at the pool level.
We formalize this pool-safety problem as finding the largest $\varepsilon$-close subset of the Cartesian product of per-component implementation choices, reducing pairwise indistinguishability under an additive utility model to a densest-window query over a sum-set.
We give four algorithms spanning the scalability spectrum---full enumeration, meet-in-the-middle, FFT convolution, and Monte Carlo sampling---covering configuration spaces from tens to $10^{38}$.
We then measure the anonymity gap end-to-end on two production cloud case studies, and find that a component's latency differences do not survive deployment unchanged: a four-runtime serverless rotation, where nothing else masks the interpreter, collapses from four-way to three-way anonymity against a VPC-adjacent adversary, while a $27$-configuration three-tier stack, where the same interpreter differences are instead absorbed by a shared $8$~ms database round-trip, delivers nine-way effective anonymity.
The framework and the two case studies together suggest a diagnostic for MTD design: rotating a component adds anonymity only if the latency differences among its variants are too small for the adversary to identify.
\end{abstract}

\section{Introduction}

Moving Target Defense (MTD) rotates a system among diverse configurations to invalidate an attacker's reconnaissance \citep{jajodia2011moving,sengupta2020survey}. Research on MTD has matured along three axes---which configurations to deploy, how to switch, and when to trigger transitions---but consistently assumes that the candidate pool itself is safe to cycle among. When configurations differ in observable behavior this assumption fails: a passive adversary who measures response latency can identify the active choice without exploiting the service, reading memory, or breaking authentication \citep{okhravi2016moving,seibert2014information}. Unfortunately approaches like random jitter does not close the gap, because a patient adversary averages over $n$ observations to reduce noise by a factor of $\sqrt{n}$ and eventually overcomes any fixed noise budget \citep{crosby2009opportunities}. Even calibrating that noise in the style of differential privacy \citep{dwork2006differential} does not fix the problem: the guarantee rests on a finite privacy budget, each probe the adversary issues consumes a portion of it, and an adversary free to keep probing eventually exhausts any budget the defender sets.

We test the pool-safety assumption on two cloud rotations and find it violated in both. In one, a four-runtime serverless rotation collapses from four-way to three-way anonymity against a VPC-adjacent adversary, because one Python interpreter version is $0.378$~ms slower than the other three. In the other, a three-tier stack of one web server, one Python runtime, and one database, each drawn from three candidates, presents $27$ configurations in total. Yet, it delivers only nine-way effective anonymity against an adversary an order of magnitude blunter than the VPC jitter floor, and three-way at the floor itself, because the database axis's client-protocol latency gaps are resolved outright while the three interpreter versions---their differences masked behind a shared $8$~ms database round-trip---are the only levels the adversary cannot separate. Both collapses point to a common lesson: nominal configuration count is not, by itself, a security metric.

To characterize this gap in general and enforce it at deployment time, we formulate the pool-safety problem as follows. A system is a $k$-component tuple $c = (s_1, \ldots, s_k) \in S_1 \times \cdots \times S_k$, and we model its observable latency as a weighted additive sum of per-component contributions:
\begin{equation}\label{eq:utility}
  U(c) = \sum_{i=1}^{k} w_i \cdot u(s_i).
\end{equation}
We then seek the largest subset $C \subseteq S_1 \times \cdots \times S_k$ satisfying $\max_{c \in C} U(c) - \min_{c \in C} U(c) \leq \varepsilon$; we call such a subset \emph{$\varepsilon$-close}, and it is the pool of configurations that an adversary of precision $\varepsilon$ cannot distinguish. Under the additive assumption the search reduces to finding the densest $\varepsilon$-width interval over the multiset $\{U(c) : c \in \prod_i S_i\}$. Although latency is our running example, any additive quantitative channel, such as CPU footprint, memory pressure, response size, energy, can fit the same formulation (\S\ref{sec:multichannel}).

We note that MTD candidate pools grow multiplicatively across rotation axes. Hence, a defense-in-depth deployment that diversifies at every layer of the request path---runtime versions, database backends, interface protocols, TLS libraries, container base images, request middleware, and connection-pool sizing---readily produces a dozen rotation axes, and four variants per axis already puts $\prod_i n_i$ at $4^{12} \approx 10^{7}$; per-service diversification in a multi-service architecture pushes it well past $10^{15}$. No single algorithm can handle this range. Thus, we develop four algorithms that together span the scalability spectrum: exact enumeration on tractable product spaces; meet-in-the-middle to extend the exact regime by a square-root factor; FFT convolution to reach product sizes the exact methods cannot address; and Monte Carlo sampling for the remainder, with a uniform-convergence density bound that remains valid even when the target window is selected from the samples.

Together, the framework and the two cloud case studies yield a diagnostic criterion for MTD design: a rotation axis converts its nominal multiplicity into anonymity only where its per-level latency gaps fall below the adversary's precision, and whether they do is a property of the deployment rather than of the axis. The same interpreter axis that leaks its odd member in the serverless rotation, where nothing masks it, supplies the only fine-precision anonymity in the three-tier stack, where a shared I/O floor hides it. The rest of the paper develops this program in three parts. \S\ref{sec:approach} formalizes the problem and presents the four algorithms. \S\ref{sec:experiments} measures the anonymity gap on cloud infrastructure and evaluates the algorithms across scales. \S\ref{sec:relatedwork} places the work in the context of prior MTD and side-channel literature.

\section{Approach}\label{sec:approach}

\subsection{Problem reduction}

The pairwise indistinguishability constraint reduces to a single range condition on $C$:

\begin{proposition}\label{prop:reduction}
The constraint $|U(c_1) - U(c_2)| \leq \varepsilon$ for all $c_1, c_2 \in C$ holds if and only if $\max_{c \in C} U(c) - \min_{c \in C} U(c) \leq \varepsilon$.
\end{proposition}

By \autoref{prop:reduction}, any maximal $\varepsilon$-close subset is the preimage under $U$ of some interval $[\ell, \ell + \varepsilon]$ on the real line, and the largest such subset is realized by the densest interval over the multiset $\{U(c) : c \in S_1 \times \cdots \times S_k\}$. The additive decomposition~\eqref{eq:utility} exposes this multiset as a $k$-fold sum-set of per-component contributions $v_i = w_i \cdot u(s_i)$. In this section, we will describe algorithms below that exploit this sum-set structure.

\subsection{Per-axis budget and visibility regimes}
\label{sec:regimes}

The additive form of $U$ decomposes not only across configurations but across the axes that select them. When the pool is a product $C = C_1 \times \cdots \times C_k$ of per-axis choices, the following corollary of \autoref{prop:reduction} shows that its total latency spread is exactly the sum of each axis's own spread, turning $\varepsilon$ into a budget the $k$ axes draw from jointly:

\begin{corollary}[Per-Axis Budget]\label{cor:budget}
Let $C = C_1 \times \cdots \times C_k$ with $C_i \subseteq S_i$, and let $\sigma_i(C_i) = \max_{s \in C_i} w_i\, u(s) - \min_{s \in C_i} w_i\, u(s)$ denote the latency spread of the selected levels on axis $i$. Then $C$ is $\varepsilon$-close if and only if $\sum_{i=1}^{k} \sigma_i(C_i) \leq \varepsilon$.
\end{corollary}

As discussed in the introduction, the two case studies illustrate what the budget view predicts. The serverless rotation's $0.378$~ms gap separates one Python interpreter version from the other three (\S\ref{sec:runtime}): the axis exceeds its budget, so it splits, and three levels survive where four were nominally available. The three-tier stack's Python axis measures a spread of zero (\S\ref{sec:fullstack}), because a shared $8$~ms database round-trip absorbs any interpreter-level difference: the same interpreter axis that split in the serverless rotation spends nothing here, so all three of its levels survive for free. Both cases point to the same lesson: an axis's regime is a property of where it is measured, not of what it is.

\autoref{cor:budget} makes this precise by sorting an axis into one of three regimes, indexed by how much of the full axis $S_i$---with spread $\sigma_i \equiv \sigma_i(S_i)$---an adversary of precision $\varepsilon$ leaves standing. If $\sigma_i \leq \varepsilon$, the axis is \emph{invisible}: every level fits inside a single $\varepsilon$-window, so its full multiplicity survives at a budget cost of only $\sigma_i$, and a degenerate axis ($\sigma_i = 0$, as the interpreter axis is in the three-tier stack) survives for free. If $\sigma_i > \varepsilon$ but some levels sit closer to each other than to the rest, the axis is \emph{partially visible}: only the densest cluster of levels---spanning at most the remaining budget---survives, and the rest are exposed, exactly as three of the serverless rotation's four interpreters survive while the fourth is isolated. If every pair of levels on the axis is separated by more than $\varepsilon$, the axis is \emph{fully visible}: no two of its levels can coexist in any $\varepsilon$-close product pool, the regime the database axis of \S\ref{sec:fullstack} occupies at any precision tighter than its narrowest adjacent gap. For $m$ evenly spaced levels, full visibility begins once the spread exceeds $(m-1)\,\varepsilon$, since the smallest adjacent gap is the binding constraint.

These regimes classify each axis on its own, which is what makes them a useful design diagnostic: an engineer can read per-axis spreads off a measurement like \Cref{fig:fullstack}(a) and predict how much of an axis's diversity survives without running any algorithm. They are not a complete characterization of the optimal pool, however, because a product-set assumption underlies them: the densest $\varepsilon$-window found directly over the sum-set can admit configurations whose per-axis differences cancel across axes, so the true optimum need not decompose as $C_1 \times \cdots \times C_k$ at all. The four algorithms of the next subsection search the sum-set itself and so capture the cross-axis cancellation that per-axis reasoning misses.

\subsection{Four algorithms across the scalability spectrum}

MTD deployments span many orders of magnitude in configuration-space size, and no single algorithm handles the full range. We therefore present four (\autoref{tab:complexity}), each dominating a distinct regime; step-by-step derivations, pseudocode, and error analysis appear in Appendix~\ref{sec:algorithms-app}.

\begin{table}[h!]
\centering
\caption{Complexity of the four algorithms. $N = \prod_i n_i$ is the product space size, $R$ is the discretized utility range, and $M$ is the sample count.}
\label{tab:complexity}
\vspace{0.5em}
\begin{tabular}{@{}llll@{}}
\toprule
Algorithm & Time & Space & Exact \\
\midrule
Full Enumeration & $O(N \log N)$ & $O(N)$ & Yes \\
Meet in the Middle & $O(N \log N)$ & $O(\sqrt{N})$ enum $+$ $O(N)$ merge & Yes \\
FFT Convolution & $O(k \cdot R \log R)$ & $O(R)$ & Approx. \\
Sampling & $O(M \log M)$ & $O(M)$ & Approx. \\
\bottomrule
\end{tabular}
\end{table}

\paragraph{Full enumeration.}
For product spaces of tractable size (say, $N \leq 10^{7}$), the utility multiset $\{U(c)\}$ can be materialized directly. We build it iteratively as a running sum-set over the per-component contributions, sort the resulting array, and apply a two-pointer sliding window of width $\varepsilon$ to identify the densest interval in a single linear pass. The algorithm returns the exact optimum in $O(N \log N)$ time, dominated by the sort.

\paragraph{Meet-in-the-middle.}
When $N$ exceeds direct enumeration but $k$ remains moderate, we extend the exact regime by partitioning the $k$ sets into two halves and enumerating each half's sum-set independently at $O(\sqrt{N})$ cost. Forming all pairwise sums between the halves reconstructs the full utility multiset, on which the same sliding window returns the exact optimum. The gain is in memory: neither half's enumeration ever exceeds $O(\sqrt{N})$, so the reachable regime grows by a square-root factor before the merge \citep{horowitz1974computing}.

\paragraph{FFT convolution.}
Once $N$ exceeds any exact method, we shift to a distributional representation of the same object. Per-component utilities are discretized into $R$ bins of width $\delta$, and the $k$ resulting histograms are combined by iterated FFT convolution to yield the histogram of $\{U(c)\}$ over the Cartesian product; a prefix-sum sliding window over that histogram identifies the densest $\varepsilon$-interval \citep{cooley1965algorithm,cormen2009introduction}. Because each per-component contribution is rounded by at most $\delta/2$, the total utility of any configuration shifts by at most $k\delta/2$, and the resulting relative count error is $O(k\delta / \varepsilon)$. The error can be driven down by shrinking $\delta$ at proportional cost in bins, while the runtime $O(k \cdot R \log R)$ depends on the utility range rather than $N$ and remains tractable well beyond the exact regime.

\paragraph{Monte Carlo sampling.}
When the discretized range $R$ itself becomes prohibitive---typically once $k$ passes roughly $15$ with wide per-component ranges---we retreat to random sampling. We draw $M$ configurations i.i.d.\ uniformly from the Cartesian product, compute and sort their utilities, and apply the sliding window to the resulting array; scaling the winning window's sample count by $N/M$ estimates the true count. The runtime $O(M \log M)$ is independent of both $N$ and $R$, but the approximate output requires a probabilistic guarantee, which we derive next.

\subsection{A uniform guarantee for Monte Carlo sampling}

Reporting the winning window's empirical density as an estimate of its true density is optimistic: this is a selection-bias problem. The window $\hat\ell$ is not chosen in advance; it is selected from the same $M$ samples used to estimate its density, as whichever candidate window happens to look densest. A window selected for looking densest is more likely than a fixed window to have benefited from a favorable sampling fluctuation, so its apparent density systematically overstates the true density. A pointwise Hoeffding bound does not fix this: it only guarantees accuracy for a window chosen before the data are drawn, and $\hat\ell$ is chosen after. We therefore need a bound that holds uniformly over every candidate window at once, so that whichever window ends up selected---regardless of how the selection rule picks it---is guaranteed accurate. The Dvoretzky--Kiefer--Wolfowitz (DKW) inequality \citep{dvoretzky1956asymptotic}, in the tight form due to \citet{massart1990tight}, supplies exactly this:

\begin{proposition}[Uniform Window-Density Bound]\label{prop:sampling}
For each $\ell \in \mathbb{R}$, let $p(\ell) = \Pr_{c}[U(c) \in [\ell, \ell+\varepsilon]]$ over a uniformly random configuration, and let $\hat p_M(\ell)$ be the empirical estimate from $M$ i.i.d.\ samples.
By the Dvoretzky--Kiefer--Wolfowitz inequality, with probability at least $1-\alpha$,
\[
  \sup_{\ell \in \mathbb{R}}\,\bigl|\hat p_M(\ell) - p(\ell)\bigr| \;\leq\; 2t,
  \qquad t = \sqrt{\tfrac{1}{2M}\log\tfrac{2}{\alpha}}.
\]
In particular, $M = O\!\left(\varepsilon_{\text{stat}}^{-2}\log(1/\alpha)\right)$ samples suffice for additive accuracy $\varepsilon_{\text{stat}}$ \emph{simultaneously} over all candidate windows---including any window selected post-hoc from the samples.
\end{proposition}

Applying this uniform bound at both $\hat\ell$ and $\ell^*$ and chaining the two inequalities via the definition of $\hat\ell$ as the empirical argmax converts the sampling error into a bound on the gap between the true density $p(\hat\ell)$ and the optimum $p^*$:

\begin{corollary}[Selection-Bias-Aware Near-Optimality]\label{cor:selection}
Let $\hat\ell = \arg\max_\ell \hat p_M(\ell)$ be the algorithm's empirical choice, $\ell^* = \arg\max_\ell p(\ell)$ the true optimum, and $p^* = p(\ell^*)$.
Under the conditions of \autoref{prop:sampling}, with probability at least $1-\alpha$,
\[
  p^* - p(\hat\ell) \;\leq\; 4t \;=\; 4\sqrt{\tfrac{1}{2M}\log\tfrac{2}{\alpha}}.
\]
That is, the \emph{true} density of the empirically-selected window is within $4t$ of the optimum, even though the same samples are used both to choose the window and to estimate its density.
\end{corollary}

The corollary's $4t$ bound justifies the use of Monte Carlo sampling in the empirical evaluation of \S\ref{sec:experiments}: with modest $M$, the true density of the empirically-selected window is within an $M^{-1/2}$ margin of the optimum. Proofs of the propositions and corollary appear in Appendix~\ref{sec:proofs}. Appendix~\ref{sec:selbias-app} demonstrates that the $4t$ constant is conservative by one to two orders of magnitude in practice.

\section{Experiments}\label{sec:experiments}

The evaluation proceeds in three passes. Two cloud case studies (\S\ref{sec:cloud-setup}--\S\ref{sec:fullstack}) anchor the paper's central finding; synthetic evaluation (\S\ref{sec:scale}--\S\ref{sec:eps-sweep}) demonstrates that the algorithm suite of \S\ref{sec:approach} scales far beyond the tested deployment sizes and that the observed Pareto shape generalizes. Then, two robustness sections (\S\ref{sec:multichannel}--\S\ref{sec:noise}) probe the framework against a channel-pivoting adversary and noisy benchmarks. Finally, we highlight the scope and limitations of our experiments in \S\ref{sec:scope}.

\subsection{Cloud setup and threat model}
\label{sec:cloud-setup}

We evaluate whether an MTD candidate pool delivers the indistinguishability it nominally offers, in a deployment where a defender rotates among a known configuration set $S$ and a passive adversary co-located in the same cloud region measures response latency. The adversary issues constant-rate HTTP probes and attempts to identify the active configuration; it does not exploit the service, read memory, or observe any signal other than end-to-end latency. Other quantitative channels---CPU footprint, memory pressure, response size, energy---are covered by the framework (\S\ref{sec:multichannel}) but not measured on the cloud setup.

The adversary's discriminatory power is captured by a single parameter, $\varepsilon$: the smallest latency difference it can reliably resolve. Smaller $\varepsilon$ denotes a sharper adversary. We treat $\varepsilon$ as the free axis of the evaluation and report every result as a function of it, so a defender can read off the operating point matching their threat model rather than commit to a single estimate.

The $\varepsilon$-close guarantee is per-observation, and what repeated observation buys the adversary is governed by the rotation dynamics. Probing one active configuration $n$ times shrinks the resolvable latency difference toward $\sigma / \sqrt{n}$, where $\sigma$ is the path jitter (\S1), but rotation caps the accumulation: a configuration that stays active for dwell time $\Delta$ against probe rate $r$ yields at most $n = r\Delta$ samples, so the adversary's achievable precision is bounded below by roughly $\sigma / \sqrt{r\Delta}$, and a defender should calibrate $\varepsilon$ to this dwell-limited floor rather than to the single-probe jitter. Across rotations the schedule itself must not leak, and prior MTD work supplies exactly this ingredient: randomized switching strategies---game-theoretic, multi-stage, and/or learned over repeated interaction \citep{sengupta2017game,sengupta2020survey}---denying the adversary a predictable observation window. Pool safety composes with these policies rather than replacing them: the switching strategy governs when and to what the system moves, while the $\varepsilon$-close pool guarantees that whichever member is active, its latency signature stays within $\varepsilon$ of every other member's.

From per-configuration latency samples the defender fits an additive ordinary-least-squares model $U(s) = \sum_{j=1}^{k} w_j[\ell_j] + \eta_s$, where $w_j[\ell]$ is the marginal contribution of component $j$ at level $\ell$ and $\eta_s$ is the residual. The fitted model induces the $\varepsilon$-close pool $C(\tau, \varepsilon) = \{s \in S : |\hat{U}(s) - \tau| \leq \varepsilon\}$: configurations whose predicted latency falls within $\varepsilon$ of an observed target $\tau$. Its cardinality $|C(\varepsilon)|$ is the adversary's residual anonymity set: large means strong MTD; small means the adversary has nearly localized the active configuration.

\paragraph{Instantiation on AWS.}
We realize this setup on AWS in the \texttt{us-east-1} region. The adversary runs on a compromised \texttt{t3.small} instance in one availability zone; the defender operates in another. Inter-AZ round-trip jitter is approximately $0.1$~ms and sets the practical floor on $\varepsilon$ for this environment. Serverless configurations (\S\ref{sec:runtime}) are exercised via 500 sequential warm invocations per Lambda through \texttt{curl}. Three-tier configurations (\S\ref{sec:fullstack}) are probed by \texttt{wrk2}\footnote{\url{https://github.com/giltene/wrk2}} at a fixed target rate; \texttt{wrk2} back-corrects each sample for scheduling delay, so the recorded latency reflects the server rather than the client's queue. Each three-tier request serves \texttt{GET /bench}---a 1\,kB response after a fixed CPU operation---and executes a 100-row \texttt{SELECT} against an RDS-hosted database. Per-configuration latencies are averaged before OLS fitting.

\subsection{Single-axis rotation: serverless runtime}
\label{sec:runtime}

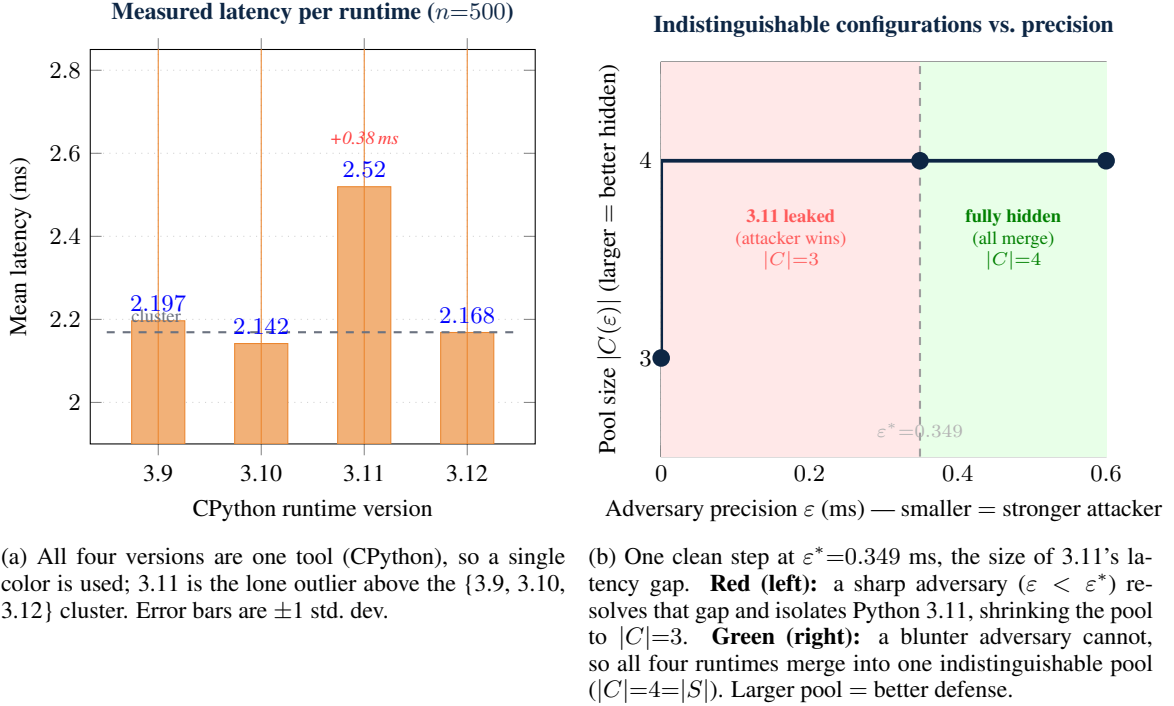
\begin{figure}[ht!]
\centering
\begin{subfigure}[t]{0.49\linewidth}
\begin{tikzpicture}
\begin{axis}[panel, ybar, bar width=20pt,
  ylabel={Mean latency (ms)}, xlabel={CPython runtime version},
  xtick={1,2,3,4}, xticklabels={3.9, 3.10, 3.11, 3.12},
  ymin=1.9, ymax=2.85, ytick={2.0,2.2,2.4,2.6,2.8},
  nodes near coords,
  nodes near coords style={font=\footnotesize, /pgf/number format/.cd, fixed, precision=3},
  title={Measured latency per runtime ($n{=}500$)},
  enlarge x limits=0.22,]
  \addplot+[fill=Clam!55, draw=Clam!85, error bars/.cd, y dir=both, y explicit,]
    coordinates {(1,2.1966)+-(0,1.409) (2,2.1418)+-(0,1.165)
                 (3,2.5195)+-(0,1.361) (4,2.1684)+-(0,1.282)};
  \draw[dashed, NeutralGray, thick]
    (axis cs:0.5,2.169)--(axis cs:4.5,2.169)
    node[pos=0.12, above, font=\scriptsize, text=NeutralGray]{cluster};
  \node[font=\scriptsize\itshape, text=red!75, anchor=south] at (axis cs:3,2.60){+0.38\,ms};
\end{axis}
\end{tikzpicture}
\caption{All four versions are one tool (CPython), so a single color is used; 3.11 is the
  lone outlier above the \{3.9, 3.10, 3.12\} cluster. Error bars are $\pm$1 std.\ dev.}
\end{subfigure}
\hfill
\begin{subfigure}[t]{0.49\linewidth}
\begin{tikzpicture}
\begin{axis}[panel,
  xlabel={Adversary precision $\varepsilon$ (ms) --- smaller $=$ stronger attacker},
  ylabel={Pool size $|C(\varepsilon)|$ (larger $=$ better hidden)},
  xmin=0, xmax=0.6, ymin=2.5, ymax=4.5, ytick={3,4},
  yticklabel style={font=\small},
  title={Indistinguishable configurations vs.\ precision},]
  \fill[red!9]   (axis cs:0,2.5)     rectangle (axis cs:0.349,4.5);
  \fill[green!9] (axis cs:0.349,2.5) rectangle (axis cs:0.6,4.5);
  \node[font=\scriptsize, text=red!65,   align=center] at (axis cs:0.175,3.6)
    {\textbf{3.11 leaked}\\(attacker wins)\\$|C|{=}3$};
  \node[font=\scriptsize, text=green!50!black, align=center] at (axis cs:0.475,3.6)
    {\textbf{fully hidden}\\(all merge)\\$|C|{=}4$};
  \addplot[pool] coordinates {(0,3)(0.349,4)(0.6,4)};
  \draw[dashed, gray!75, line width=0.8pt]
    (axis cs:0.349,2.5)--(axis cs:0.349,4.5);
  \node[font=\scriptsize, text=gray!60, anchor=south] at (axis cs:0.349,2.55)
    {$\varepsilon^*{=}0.349$};
\end{axis}
\end{tikzpicture}
\caption{One clean step at $\varepsilon^{*}{=}0.349$~ms, the size of 3.11's
  latency gap. \textbf{Red (left):} a sharp adversary ($\varepsilon<\varepsilon^{*}$)
  resolves that gap and isolates Python~3.11, shrinking the pool to $|C|{=}3$.
  \textbf{Green (right):} a blunter adversary cannot, so all four runtimes merge into
  one indistinguishable pool ($|C|{=}4{=}|S|$). Larger pool $=$ better defense.}
\end{subfigure}
\caption{\textbf{Serverless Runtime Rotation ($k{=}1$, $|S|{=}4$).}
  A single active axis with four distinct levels yields a Pareto step.
  The attack works exactly as the model predicts, and all four pool-discovery
  algorithms agree on the step location.}
\label{fig:runtime}
\end{figure}

The single-axis instance rotates four AWS Lambda functions running distinct CPython versions (3.9--3.12) behind the same request handler. The handler performs no blocking I/O, so the measured latency reflects interpreter behavior directly. \Cref{fig:runtime}(a) shows the mean latency of each runtime over 500 warm invocations: CPython~3.9, 3.10, and 3.12 fall within $0.06$~ms of one another ($2.142$--$2.197$~ms), while CPython~3.11 sits at $2.520$~ms---a $0.378$~ms gap above the cluster.

The Pareto curve (\cref{fig:runtime}(b)) has a single step at $\varepsilon^{*} = 0.349$~ms, the size of the 3.11 gap. Against an adversary of precision $\varepsilon < \varepsilon^{*}$ (red region), the pool is $|C| = 3$: CPython~3.11 is unconditionally identified while the other three runtimes remain indistinguishable. Against a blunter adversary, say one who resides in a different cloud region with cross-region access (i.e., $\varepsilon \geq \varepsilon^{*}$, green region), all four merge into a single pool of $|C| = 4$. All four pool-discovery algorithms of \S\ref{sec:approach} return the same step location, confirming it is exact and not a discretization artifact; full enumeration over $|S| = 4$ completes in roughly $5\,\mu$s.

The direction of the 3.11 gap is worth pausing on: published interpreter benchmarks generally place CPython~3.11 \emph{ahead} of its neighbors on CPU-bound workloads, so a reader familiar with those benchmarks may expect the opposite result. The gap itself is not an artifact of noise: per-invocation jitter gives each runtime a standard deviation of $1.2$--$1.4$~ms, but at $n = 500$ the standard error of each mean is roughly $0.06$~ms, so the $0.378$~ms gap separates the means by more than four standard errors of their difference, and all invocations are warm, ruling out cold-start initialization as an explanation.
What the measurement cannot isolate is why the direction reverses here: the deployed artifact couples the interpreter with a version-specific Lambda runtime image, and either the adaptive interpreter's behavior on this particular handler or the runtime layer itself may be responsible. The distinction does not matter for our argument: we attribute the gap to the deployed configuration rather than to CPython~3.11 in general, and the conclusions that follow depend only on the deployed configurations differing measurably, not on why they differ. Had 3.10 been the outlier instead of 3.11, the framework would report the identical $|C|=3$ split at the identical threshold, because it operates on measured latency values, not on which runtime produced them.

The takeaway is narrow but concrete: a defender who deploys four-runtime rotation expects four-way anonymity but obtains three-way anonymity against a VPC-adjacent adversary, and CPython~3.11 leaks its identity unconditionally. The results also show that even runtime is expected to be faster, yet the gap is a stable fingerprint that persists across all 500 invocations that can be resolved by the adversary. The absence of blocking I/O is the condition that lets the interpreter-level signal reach the wire; the next case study shows what happens when this condition fails. 

\subsection{Full-stack rotation: three-tier deployment}
\label{sec:fullstack}

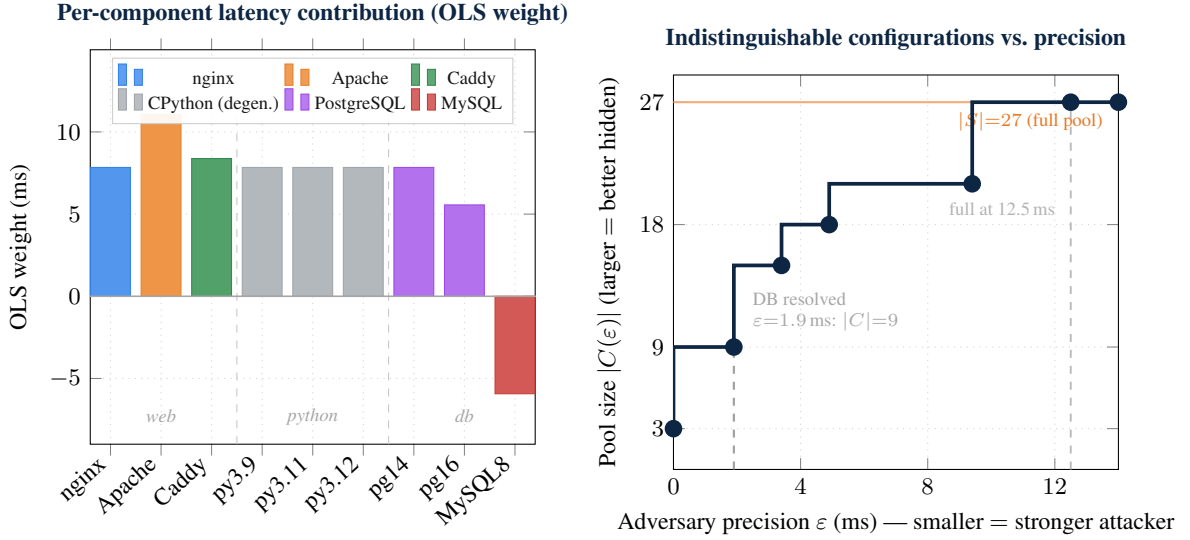
\begin{figure}[ht!]
\centering
\begin{subfigure}[t]{0.49\linewidth}
\begin{tikzpicture}
\begin{axis}[panel, ybar, bar width=15pt,
  ylabel={OLS weight (ms)},
  xtick={1,2,3,4,5,6,7,8,9},
  xticklabels={nginx,Apache,Caddy,py3.9,py3.11,py3.12,pg14,pg16,MySQL8},
  x tick label style={rotate=42, anchor=east, font=\footnotesize},
  ymin=-9, ymax=15, ytick={-5,0,5,10},
  enlarge x limits=0.05,
  legend style={at={(0.5,0.98)}, anchor=north, legend columns=3,
                font=\scriptsize, draw=gray!40, inner sep=2pt},
  title={Per-component latency contribution (OLS weight)},]
  \addplot[ybar, bar shift=0pt, fill=Cnginx!75,  draw=Cnginx!90] coordinates {(1,7.84)};
    \addlegendentry{nginx}
  \addplot[ybar, bar shift=0pt, fill=Capache!75, draw=Capache!90] coordinates {(2,11.06)};
    \addlegendentry{Apache}
  \addplot[ybar, bar shift=0pt, fill=Ccaddy!75,  draw=Ccaddy!90] coordinates {(3,8.38)};
    \addlegendentry{Caddy}
  \addplot[ybar, bar shift=0pt, fill=Cpyv!75,    draw=Cpyv!95] coordinates {(4,7.84)(5,7.84)(6,7.84)};
    \addlegendentry{CPython (degen.)}
  \addplot[ybar, bar shift=0pt, fill=Cpg!70,     draw=Cpg!90] coordinates {(7,7.84)(8,5.56)};
    \addlegendentry{PostgreSQL}
  \addplot[ybar, bar shift=0pt, fill=Cmysql!75,  draw=Cmysql!90] coordinates {(9,-5.93)};
    \addlegendentry{MySQL}
  \draw[dashed,gray!45] (axis cs:3.5,-9)--(axis cs:3.5,11);
  \draw[dashed,gray!45] (axis cs:6.5,-9)--(axis cs:6.5,11);
  \draw[gray!70,line width=0.6pt] (axis cs:0.4,0)--(axis cs:9.6,0);
  \node[font=\scriptsize\itshape,text=gray!70] at (axis cs:2,-7.3){web};
  \node[font=\scriptsize\itshape,text=gray!70] at (axis cs:5,-7.3){python};
  \node[font=\scriptsize\itshape,text=gray!70] at (axis cs:8,-7.3){db};
\end{axis}
\end{tikzpicture}
\caption{Each web server (nginx/Apache/Caddy) and database (PostgreSQL/MySQL) is drawn in
  its own color. The DB axis dominates (13.8\,ms spread); the three CPython versions are
  degenerate (identical gray bars). MySQL's negative weight is its advantage over the
  pg14 reference.}
\end{subfigure}
\hfill
\begin{subfigure}[t]{0.49\linewidth}
\begin{tikzpicture}
\begin{axis}[panel,
  xlabel={Adversary precision $\varepsilon$ (ms) --- smaller $=$ stronger attacker},
  ylabel={Pool size $|C(\varepsilon)|$ (larger $=$ better hidden)},
  xmin=0, xmax=14, ymin=0, ymax=29, ytick={3,9,18,27}, xtick={0,4,8,12},
  title={Indistinguishable configurations vs.\ precision},]
  \draw[AccentOrange!60, line width=0.8pt] (axis cs:0,27)--(axis cs:14,27);
  \node[font=\scriptsize, text=AccentOrange, anchor=north east] at (axis cs:13.8,27)
    {$|S|{=}27$ (full pool)};
  \addplot[pool] coordinates {(0,3)(1.9,9)(3.4,15)(4.9,18)(9.4,21)(12.5,27)(14,27)};
  \draw[dashed, gray!75, line width=0.8pt] (axis cs:1.9,0)--(axis cs:1.9,9);
  \node[font=\scriptsize, text=gray!70, anchor=south west, align=left]
    at (axis cs:2.2,9.5) {DB resolved\\$\varepsilon{=}1.9$\,ms: $|C|{=}9$};
  \draw[dashed, gray!55, line width=0.7pt] (axis cs:12.5,0)--(axis cs:12.5,27);
  \node[font=\scriptsize, text=gray!60, anchor=south east] at (axis cs:12.3,18)
    {full at 12.5\,ms};
\end{axis}
\end{tikzpicture}
\caption{Large steps set by the DB latency gaps. One DB group (9 web$\times$Python
  combos) merges at $\varepsilon{=}1.9$~ms; the pool saturates to $|S|{=}27$
  only at $12.5$~ms---over $30\times$ the precision Runtime Rotation needed.}
\end{subfigure}
\caption{\textbf{Full-Stack Rotation (web$\times$Python$\times$DB, $k{=}3$, $|S|{=}27$).}
  One experiment, two opposite axis behaviors: the database axis dominates with an
  order-of-magnitude larger spread than any other component, while the Python axis is
  fully degenerate behind an $8$~ms I/O floor. The degenerate axis is all that survives
  fine precision ($|C|=3$ at the VPC jitter floor); the leak is concentrated in a single
  high-value axis.}
\label{fig:fullstack}
\end{figure}

The three-tier instance rotates axes at every layer of the serving stack: each request is served by one of three web servers (nginx, Apache, Caddy), executed by one of three Python runtimes (3.9, 3.11, 3.12), and backed by a 100-row \texttt{SELECT} against one of three databases (PostgreSQL~14, PostgreSQL~16, MySQL~8). The adversary attempts to localize the active configuration in the $3 \times 3 \times 3 = 27$-point grid.

\Cref{fig:fullstack}(a) shows the fitted OLS weights. The database axis spans $13.8$~ms: MySQL~8 sits $13.8$~ms below PostgreSQL~14, with PostgreSQL~16 in between. The web-server axis spans $3.2$~ms from nginx to Apache. The Python axis spans \emph{zero}: all three interpreter weights are numerically identical. The database spread traces to the client-library protocol: the MySQL client reaches the query in fewer round-trips than the PostgreSQL driver, which negotiates types and re-authenticates. The interpreter degeneracy is structural rather than incidental: the request path blocks for approximately $8$~ms on the downstream database round-trip regardless of the interpreter version, so no interpreter-specific signal reaches the wire.

\Cref{fig:fullstack}(b) shows the Pareto curve. It climbs in large steps set by the database axis, reaching $|C| = 9$ at $\varepsilon = 1.9$~ms---the point at which one three-configuration database group has been resolved but the nine web-server $\times$ interpreter combinations within it remain indistinguishable---and saturates to $|S| = 27$ only at $\varepsilon = 12.5$~ms, more than thirty times the precision the single-axis rotation required to saturate. Meet-in-the-middle (\S\ref{sec:approach}) makes the axis diagnosis explicit: splitting on the database axis surfaces the nine-way anonymity set at $\varepsilon = 1.9$~ms, while splitting on the CPython axis yields no separating matches at any tolerance.

The two case studies instantiate the visibility regimes of \autoref{cor:budget}. In the serverless rotation the interpreter axis is partially visible at VPC-adjacent precision: the $0.378$~ms gap that isolates CPython~3.11 exceeds the $0.1$~ms jitter floor, while the $0.06$~ms spread of the remaining cluster falls below it, so three of the four levels survive as a pool. In the three-tier stack the same axis is invisible---its fitted spread is zero because the request path blocks for approximately $8$~ms on the database round-trip regardless of interpreter---so its full three-way multiplicity survives at every precision, and it is the sole reason the pool never falls below $|C| = 3$. The database axis sits at the opposite extreme: its adjacent-level gaps ($2.3$~ms and $11.5$~ms) both dwarf the jitter floor, so at fine precision its nominal three-way diversity contributes no anonymity at all, and even a much blunter $1.9$~ms adversary holds the nominal $27$ configurations to nine-way effective anonymity. The comparison yields the diagnostic in its most useful form: an axis converts nominal multiplicity into anonymity only where its per-level latency gaps fall below the adversary's precision, and whether they do is a property of the deployment, not of the axis. The interpreter axis that leaks CPython~3.11 in the serverless rotation, where nothing masks it, supplies the only fine-precision anonymity in the three-tier stack, where a shared I/O floor hides it. A defender should therefore audit each axis's level gaps against the deployment's own noise floor and I/O structure rather than count nominal variants.

\subsection{Scalability of the algorithm suite}
\label{sec:scale}

The cloud case studies exercise the framework at $k \in \{1, 3\}$ with $|S_j| \leq 4$, but realistic MTD deployments extend well beyond this scale. As discussed in the introduction, diversification across the full request path---runtime versions, database backends, interface protocols, TLS libraries, and connection-pool sizing among others---readily produces a dozen or more rotation axes, and multiplying their variants puts the configuration space near $10^{7}$ per service and beyond $10^{15}$ across a multi-service architecture. At those scales exact enumeration and meet-in-the-middle both become infeasible, leaving FFT convolution and Monte Carlo sampling as the only viable options. We therefore evaluate all four algorithms on synthetic datasets spanning $10^{1}$ to $10^{38}$ configurations, chosen to match the additive latency structure fit in \S\ref{sec:cloud-setup}.

\paragraph{Setup.}
Per-component values are drawn uniformly from $[0, 10]$ under three additive utility functions: linear (uniform weights), square-root (linear weights), and quadratic (exponential-decay weights). The product-space sizes span from $|S| = 24$ (matching \S\ref{sec:runtime}) to $|S| = 1.1 \times 10^{38}$ (\autoref{tab:datasets}). Full enumeration and meet-in-the-middle run only on $|S| \leq 2 \times 10^{6}$; FFT convolution (bin width $\delta = 0.05$) and Monte Carlo sampling ($M = 5 \times 10^{5}$) run on all five datasets.

\begin{table}[t]
\centering
\caption{Synthetic datasets used for scalability evaluation.}
\label{tab:datasets}
\vspace{0.4em}
\begin{tabular}{@{}llrrr@{}}
\toprule
Name & $k$ & Representative sizes & $\prod n_i$ & $\varepsilon$ \\
\midrule
\textsc{Tiny}   & 3  & $[3, 2, 4]$     & $2.4 \times 10^1$   & 2.0 \\
\textsc{Small}  & 5  & $[10, 8, 12, 9, 7]$ & $6.0 \times 10^4$ & 1.5 \\
\textsc{Medium} & 6  & $[12, 10, 14, 8, 11, 9]$ & $1.3 \times 10^6$ & 3.0 \\
\textsc{Large}  & 10 & $\sim\! 30$--$50$ each  & $8.4 \times 10^{15}$ & 5.0 \\
\textsc{Huge}   & 20 & $\sim\! 60$--$100$ each & $1.1 \times 10^{38}$ & 2.0 \\
\bottomrule
\end{tabular}
\end{table}

\paragraph{Findings.}
\autoref{tab:results} reports representative results. All exact algorithms return identical $|C|$ on every dataset where both apply, confirming correctness. FFT convolution completes in under $30$~ms even on the $10^{38}$-configuration dataset because its runtime depends on the discretized utility range, not the product-space size; it overcounts by up to a factor of two near bin boundaries at $\delta = 0.05$ (a $5\times$ memory increase to $\delta = 0.01$ closes the gap), but agrees with the exact count whenever the utility distribution is concentrated. Monte Carlo sampling has near-constant runtime ($\sim\!0.2$~s) governed by $M$ alone and tracks exact results within $1\%$ relative error on the datasets where exact ground truth exists. On \textsc{Large} and \textsc{Huge} the FFT and sampling estimates agree to within the FFT discretization bound, providing mutual cross-validation where no exact reference is available.

\begin{table}[t]
\centering
\caption{Selected experimental results. ``Enum'' = Full Enumeration, ``MitM'' = Meet in the Middle, ``FFT'' = FFT Convolution, ``Samp'' = Sampling. Times in seconds.}
\label{tab:results}
\vspace{0.4em}
\small
\begin{tabular}{@{}ll rr rr rr@{}}
\toprule
& & \multicolumn{2}{c}{Enum / MitM} & \multicolumn{2}{c}{FFT} & \multicolumn{2}{c}{Sampling} \\
\cmidrule(lr){3-4} \cmidrule(lr){5-6} \cmidrule(lr){7-8}
Dataset & Utility & $|C|$ & Time & $|C|$ & Time & $|C|$ & Time \\
\midrule
\textsc{Tiny}   & Linear   & 11        & 0.0001 & 24$^\dagger$  & 0.001 & 11       & 0.17 \\
\textsc{Tiny}   & Sqrt     & 24        & 0.0002 & 24            & 0.001 & 24       & 0.14 \\
\textsc{Tiny}   & Quadratic& 4         & 0.0002 & 5$^\dagger$   & 0.003 & 4        & 0.20 \\
\midrule
\textsc{Small}  & Linear   & 28,246    & 0.024  & 56,466$^\dagger$ & 0.001 & 28,293  & 0.21 \\
\textsc{Small}  & Sqrt     & 59,181    & 0.012  & 60,480        & 0.001 & 59,180   & 0.13 \\
\textsc{Small}  & Quadratic& 2,703     & 0.021  & 3,064$^\dagger$  & 0.005 & 2,713   & 0.20 \\
\midrule
\textsc{Medium} & Linear   & 1,162,488 & 0.48   & 1,330,560$^\dagger$ & 0.001 & 1,161,507 & 0.17 \\
\textsc{Medium} & Sqrt     & 1,330,560 & 0.26   & 1,330,560     & 0.001 & 1,330,560 & 0.13 \\
\textsc{Medium} & Quadratic& 176,957   & 0.46   & 188,408$^\dagger$  & 0.005 & 175,931  & 0.21 \\
\midrule
\textsc{Large}  & Linear   & ---       & ---    & $8.4 \times 10^{15}$ & 0.002 & $8.4 \times 10^{15}$ & 0.16 \\
\textsc{Large}  & Quadratic& ---       & ---    & $1.3 \times 10^{15}$ & 0.008 & $1.3 \times 10^{15}$ & 0.21 \\
\midrule
\textsc{Huge}   & Linear   & ---       & ---    & $1.0 \times 10^{38}$ & 0.004 & $9.5 \times 10^{37}$ & 0.22 \\
\textsc{Huge}   & Quadratic& ---       & ---    & $7.3 \times 10^{36}$ & 0.028 & $7.2 \times 10^{36}$ & 0.24 \\
\bottomrule
\end{tabular}

\vspace{4pt}
{\footnotesize $^\dagger$FFT overcount due to discretization at $\delta = 0.05$.}
\end{table}

\subsection{Framework generality across the \texorpdfstring{$\varepsilon$}{epsilon}-axis}
\label{sec:eps-sweep}

The step-function Pareto structure observed on the two cloud case studies is intrinsic to any additive utility over independent per-component distributions, not an artifact of the specific deployments. To confirm this we sweep $\varepsilon$ over $[0.05, 6.0]$ in $60$ steps on \textsc{Medium} (exact enumeration and FFT) and \textsc{Large} (FFT only; exact methods are infeasible at $|S| = 8.4 \times 10^{15}$), reporting the pool fraction $|C|/|S|$ per utility function in \autoref{fig:eps-sweep}.

\begin{figure}[t]
\centering
\includegraphics[width=\linewidth]{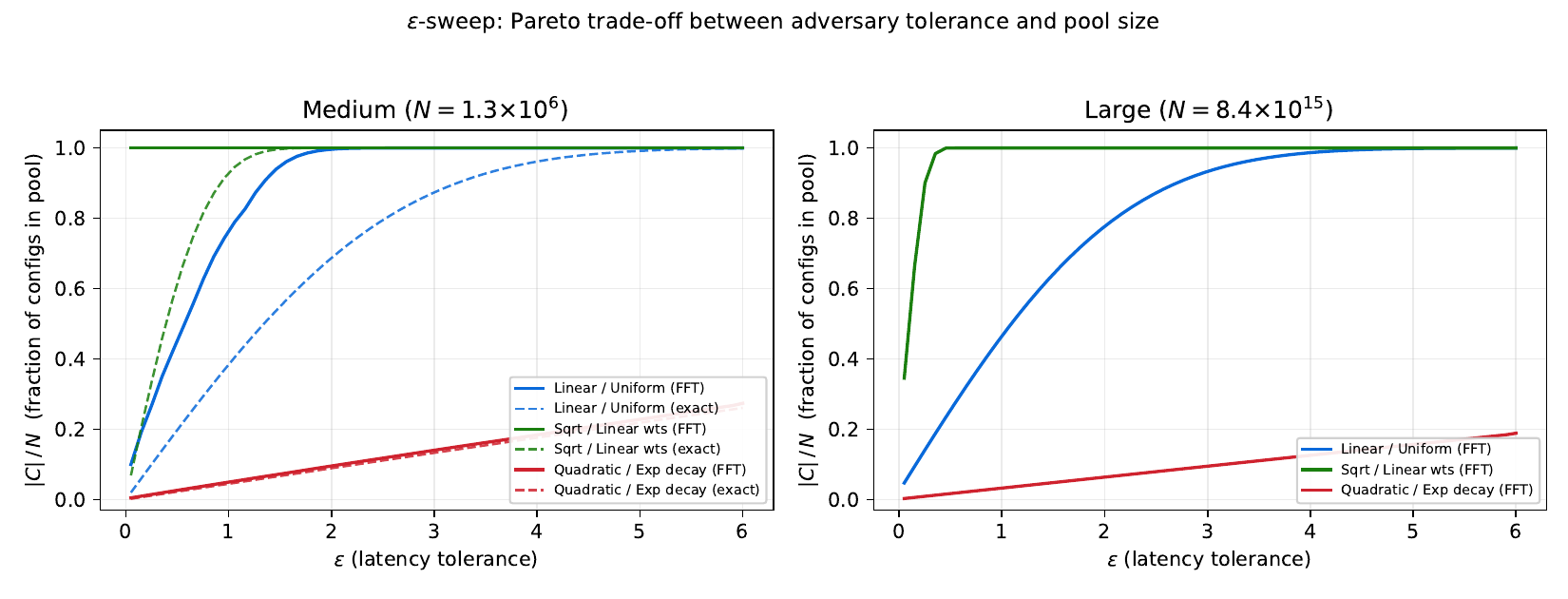}
\caption{Pareto trade-off between adversary precision $\varepsilon$ and configuration-pool fraction $|C|/|S|$.
Solid lines are FFT estimates ($\delta{=}0.05$); dashed lines are exact enumeration where tractable.
The curves reveal an intrinsic, dataset-size-independent shape determined by the utility form: concave (\emph{Sqrt}) saturates rapidly, convex (\emph{Quadratic}) admits only a small pool fraction even at $\varepsilon{=}6$, with linear in between.
A defender selects an operating point on this curve given an estimate of the deployment-path noise floor.}
\label{fig:eps-sweep}
\end{figure}

Three facts follow. First, the \textsc{Medium} and \textsc{Large} curves have near-identical shape per utility, so the Pareto profile is a function of the utility form (concave, linear, convex) rather than the dataset scale---the same shape can be expected of larger deployments constructed under the same additive assumption. Second, the FFT--exact gap on \textsc{Medium} matches the $O(k\delta/\varepsilon)$ analysis of \S\ref{sec:approach}: the Linear curve overestimates by roughly $2\times$ at small $\varepsilon$ and converges as $\varepsilon$ grows, while the Quadratic curve agrees almost everywhere because its wider utility range makes $k\delta/\varepsilon$ small. Third, the curves let a defender read the sensitivity of pool size to precision directly: on \textsc{Large} under Quadratic utility, halving $\varepsilon$ from $6$ to $3$ halves the pool fraction from $0.19$ to $0.09$, so tightening indistinguishability by $2\times$ costs a $2\times$ smaller pool.

\subsection{Multi-channel indistinguishability: a synthetic exploration}
\label{sec:multichannel}

A determined adversary need not stop at latency. Response size, CPU footprint, memory pressure, and energy draw are equally remotely-observable in many deployments, and an adversary will pivot to whichever channel is most discriminative. We argued in \S\ref{sec:approach} that the framework is metric-agnostic; here we make that quantitative. Because the cloud measurements exercise latency only (\S\ref{sec:cloud-setup}), this subsection runs on the synthetic \textsc{Medium} dataset with proxy channels---the numbers below demonstrate the multi-channel formulation, not a cloud measurement. We treat \emph{Linear/Uniform} as a latency proxy and \emph{Quadratic/Exp-decay} as a CPU-or-memory proxy, both on \textsc{Medium} at $\varepsilon = 3$. For each formulation we compute the densest indistinguishable pool: (i) channel~1 alone, (ii) channel~2 alone, (iii) the naive intersection of the two channel-optimal windows (configurations that happen to lie in both single-channel-optimal pools), and (iv) the joint formulation, which optimises $(\ell_1, \ell_2)$ over the 2-D box $[\ell_1, \ell_1 + \varepsilon] \times [\ell_2, \ell_2 + \varepsilon]$ to maximise the count of configurations in both windows simultaneously. Formulation (iv) is solved by the natural 2-D extension of the FFT histogram---a 2-D histogram on $(U_1, U_2)$ followed by a 2-D prefix-sum sliding rectangle---with the same discretization caveats as the 1-D method; we report the \emph{exact} count of configurations actually in the chosen box, which corrects the discretization overestimate.

\begin{figure}[t]
\centering
\includegraphics[width=\linewidth]{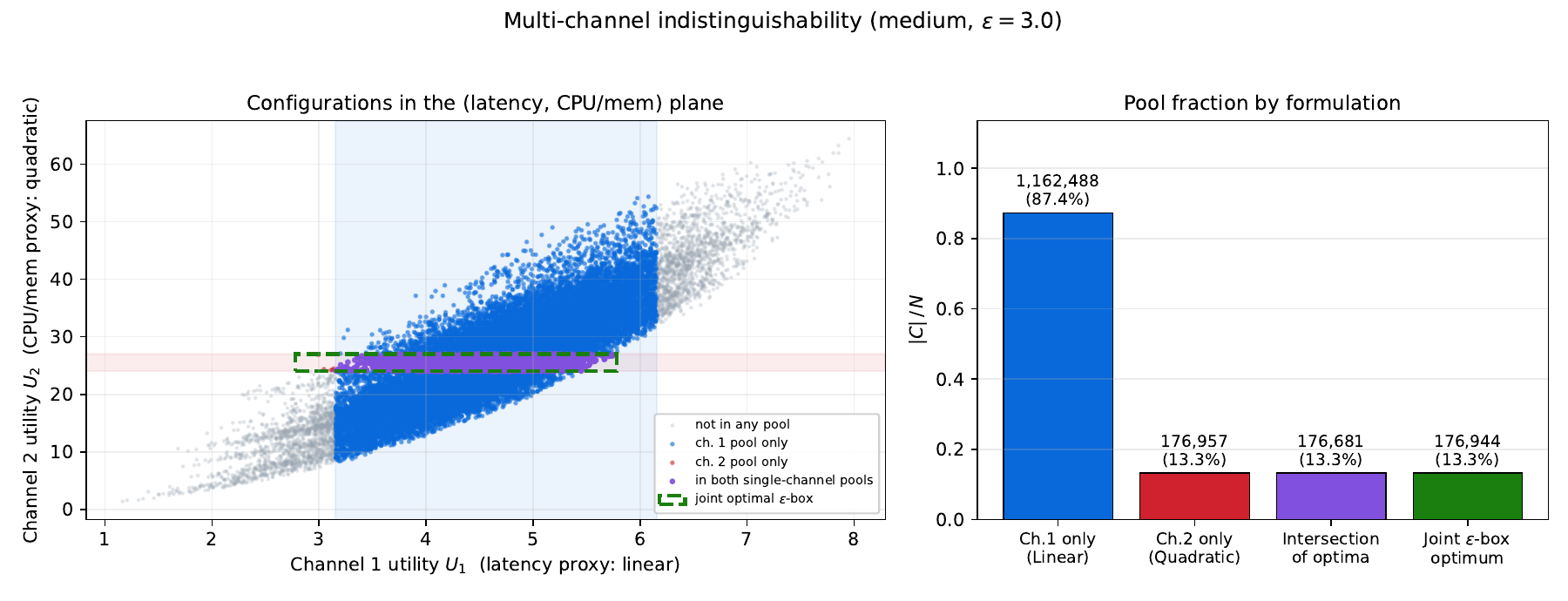}
\caption{Two-channel indistinguishability on \textsc{Medium} at $\varepsilon = 3$.
Left: scatter of 25k random configurations in the (latency proxy, CPU/memory proxy) plane.
Vertical blue band: channel-1 optimal $\varepsilon$-window.
Horizontal red band: channel-2 optimal $\varepsilon$-window.
Green dashed box: joint optimal $\varepsilon$-box.
Purple points lie in both single-channel-optimal pools.
Right: pool fraction $|C|/|S|$ for each formulation.
A defender who plans only against latency would believe they retain $87\%$ of configurations; honest accounting against \emph{both} channels admits only $13\%$.}
\label{fig:multichannel}
\end{figure}

\autoref{fig:multichannel} surfaces three facts. First, the joint pool size ($13.30\%$) is essentially equal to the narrower channel's single-channel pool ($13.30\%$): the more selective channel dominates. Second, the naive intersection ($13.28\%$) is nearly indistinguishable from the joint optimum here, indicating that on this dataset the channel-2 optimal window happens to lie almost entirely within the channel-1 optimal window---a coincidence that is not guaranteed in general but holds when channels are positively correlated, as the scatter plot shows. Third, and most consequential if the pattern carries to real deployments: a defender who calibrates $\varepsilon$ to one channel alone overstates the indistinguishable pool by $6.5\times$ on this synthetic instance ($87\% \to 13\%$). The recipe is straightforward: run the pool-selection procedure once per quantitative channel (or jointly via formulation (iv)), and treat the intersection as the true anonymity budget.

\subsection{Benchmark noise as a practical caveat}
\label{sec:noise}

The framework assumes the defender knows per-component utilities $u(s_i)$ exactly, but real deployments estimate them from benchmarks. The noise studied here is therefore the \emph{defender's}---error in the benchmark estimates of $u(s_i)$---and is distinct from the adversary's per-observation variance, which is already folded into the precision parameter $\varepsilon$ by the dwell-time argument of \S\ref{sec:cloud-setup}. To quantify the resulting exposure we add i.i.d.\ Gaussian noise $\xi_c \sim \mathcal{N}(0, \sigma^2)$ to each true configuration utility $U(c)$, run pool selection on $U_{\text{noisy}} = U + \xi$, and measure the \emph{true} latency span $\max_{c \in C} U(c) - \min_{c \in C} U(c)$ of the selected pool. Whenever this span exceeds the target $\varepsilon$, the pool is not actually $\varepsilon$-indistinguishable in deployment, regardless of what the defender's algorithm reported.

\begin{figure}[t]
\centering
\includegraphics[width=\linewidth]{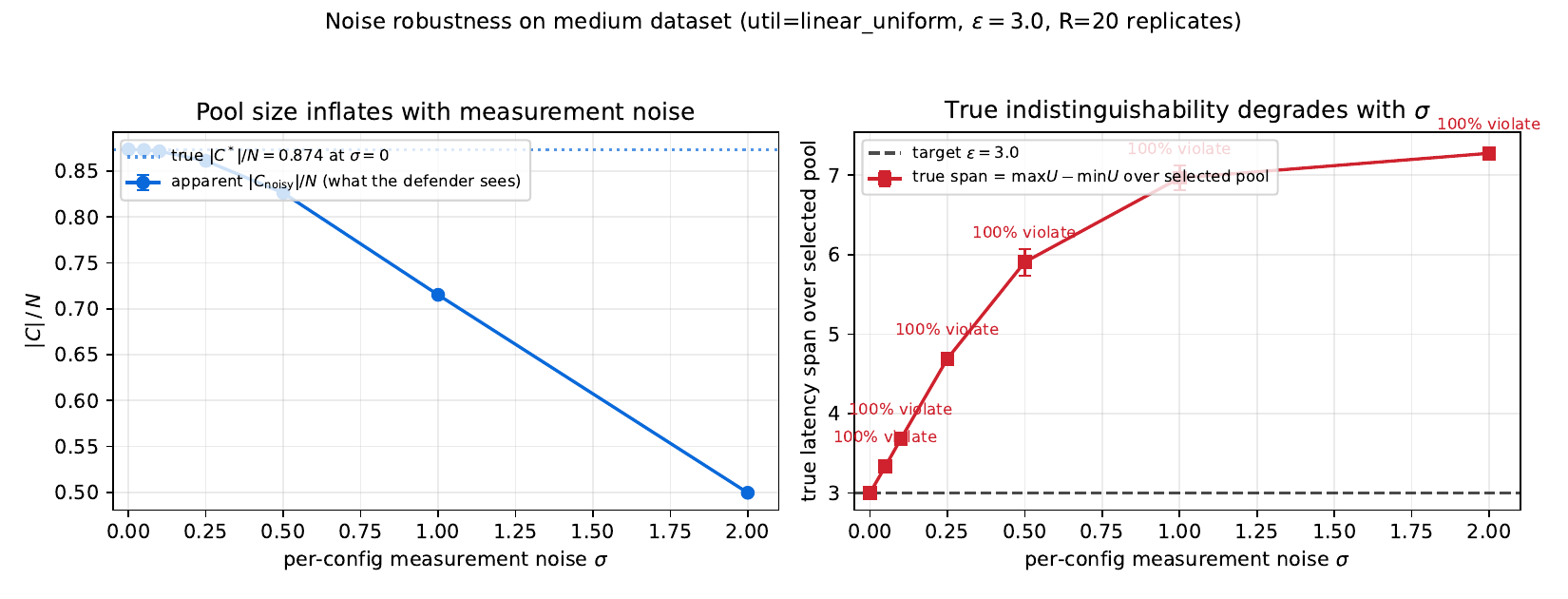}
\caption{Effect of per-configuration benchmark noise $\sigma$ on pool selection (\textsc{Medium}, Linear/Uniform, $\varepsilon = 3$, 20 replicates per $\sigma$).
Left: apparent pool fraction $|C_{\text{noisy}}|/|S|$ that the defender's algorithm reports.
Right: true latency span of that pool under the noise-free utility.
The dashed line is the target $\varepsilon$.
Even tiny noise ($\sigma = 0.05$) inflates the true span above $\varepsilon$ in $100\%$ of replicates.}
\label{fig:noise}
\end{figure}

\autoref{fig:noise} shows a sharp asymmetry between apparent and actual indistinguishability. The apparent pool size (left) is robust to modest noise: at $\sigma = 0.25$ the defender still believes they have a pool covering $86\%$ of configurations, only $1.5$ percentage points below the noise-free truth. The actual span (right) degrades immediately: at $\sigma = 0.05$, $100\%$ of replicates produce pools whose true span exceeds $\varepsilon$ (mean $3.34 > 3.0$), and at $\sigma = 0.5$ the span has doubled to $5.9$. Two mitigations follow. First, averaging benchmark measurements over $T$ trials drives the per-configuration noise variance to $\sigma_0^2 / T$, so $T = O((\sigma_0 / \sigma_{\text{tol}})^2)$ trials per configuration suffice for tolerance $\sigma_{\text{tol}}$. Second, shrinking the design tolerance from $\varepsilon$ to $\varepsilon - c\sigma$ leaves headroom that absorbs the residual span, with the constant $c$ chosen by calibration on a held-out replicate set. Either mitigation closes the apparent--true gap the figure exposes.

\subsection{Scope and limitations}
\label{sec:scope}

The cloud case studies exercise two rotation designs on a single provider, in a single region, against a passive adversary co-located in an adjacent availability zone. VPC-adjacent adversary precision was measured directly; the internet-scale precision cited for context is drawn from prior remote timing work \citep{seibert2014information} and not measured here. The threat model is passive and the candidate pool is public (\S\ref{sec:cloud-setup}); what must remain hidden is the active configuration, and the dwell-time argument of \S\ref{sec:cloud-setup} bounds what repeated passive observation achieves, but we do not model an active adversary that adaptively times probes around rotation events. Pool safety also composes with, rather than replaces, the switching policy: a schedule that is itself observable---say, fixed round-robin with a known phase---can reveal the active configuration even from an $\varepsilon$-close pool, and the randomized game-theoretic switching strategies of prior MTD work \citep{sengupta2017game,sengupta2020survey} are the natural complement. The equalization-at-source defense implied by the \S\ref{sec:fullstack} diagnostic---warming a uniform connection pool so that the interpreter axis's shared $8$~ms floor is removed---is proposed, not empirically validated; testing it is future work. Latency is the primary channel we measure on the cloud setup; other quantitative channels are covered by the framework (\S\ref{sec:multichannel}) but not exercised on real hardware. Finally, the fitted OLS model assumes additive per-component contributions; interaction effects, if present, are absorbed into the residual $\eta_s$ and would degrade predictive accuracy in a way our two case studies did not surface.

\section{Related Work}\label{sec:relatedwork}

\paragraph{Moving target defense.}
Moving Target Defense, introduced as a proactive-defense paradigm by \citet{jajodia2011moving}, has matured along three axes---which configurations to deploy, how to switch, and when to switch---systematized in \citet{sengupta2020survey}, with randomized game-theoretic strategies addressing the how and the when \citep{sengupta2017game}. A closely related diversity paradigm is N-variant execution \citep{cox2006nvariant}, in which multiple diversified variants run in parallel and disagreement between them signals an exploit; our pool-safety framework applies directly to the temporal-rotation setting and, with modification, to the parallel-execution setting, where the observable derives from all running variants rather than from a single active choice. This body of work consistently treats the candidate pool as given input and safe to cycle among; whether the pool itself leaks the active configuration through observable behavior has received limited formal treatment. Closest to our motivation, \citet{okhravi2016moving} identify information leakage via side channels as a fundamental weakness across the five domains of MTD (dynamic platforms, runtime environments, software diversification, data encoding, network properties), and report that a single observable timing value suffices to narrow down $45\%$ of diversified software functions to an uncertainty set of ten or fewer candidates; their conclusion is that MTD techniques must periodically re-randomize system internals. We take a complementary approach: we frame pool-safety as a formal indistinguishability property, quantify it on production cloud deployments, and give the algorithmic tooling to enforce it at scale.

\paragraph{Side-channel attacks on diversified systems.}
Timing side channels have been extensively studied in cryptographic implementations \citep{kocher1996timing} and cache behavior \citep{bernstein2005cache}. \citet{crosby2009opportunities} established the remote feasibility of these attacks over WAN paths, demonstrating that microsecond-scale differences can be recovered through the $\sqrt{n}$ averaging invoked in \S1. Most directly relevant to MTD, \citet{seibert2014information} demonstrate remote timing side-channel attacks that infer memory layout from diversified code without requiring any memory disclosure vulnerability, showing that timing measurements alone uniquely fingerprint up to $60\%$ of diversified functions in \texttt{libc}. Dynamic software diversity has been proposed as a defense against microarchitectural timing channels \citep{crane2015thwarting}; our contribution takes the analogous defensive move at the system-configuration granularity for network-level timing. Our threat model transposes this signal from code diversification to configuration diversification: rather than fingerprinting individual functions, the adversary identifies the active configuration among those in the MTD pool by observing response latency. The resulting notion of $|C|$-way indistinguishability is directly analogous to $k$-anonymity in privacy-preserving data release \citep{sweeney2002kanonymity}, transposed from record equivalence classes to configuration equivalence classes under an observable latency channel. We contribute both a formal indistinguishability constraint that eliminates this signal by construction---configurations in the selected pool differ in latency by at most $\varepsilon$---and end-to-end measurement of the constraint's slack on production cloud infrastructure, extending prior remote-timing evidence from code-fingerprinting granularity to system-configuration anonymity.

\paragraph{Algorithmic techniques.}
The algorithms we compose are individually classical. The two-pointer sliding window over sorted data is elementary; we apply it to a sumset that arises implicitly from a $k$-fold Cartesian product. Meet-in-the-middle originates in Horowitz and Sahni's treatment of the subset-sum problem \citep{horowitz1974computing}. FFT convolution for counting representations in a sumset is a standard use of the discrete Fourier transform \citep{cooley1965algorithm,cormen2009introduction} that we lift to $k$-fold convolution over discretized utility histograms. Random sampling of additive-utility configuration spaces has precedent in software product lines \citep{oh2017near}. The novelty of our algorithmic contribution is not in any single technique but in their composition into a suite that covers realistic MTD scales, together with the uniform-convergence analysis that makes the sampling algorithm sound when the sliding window is chosen from the samples themselves.

\section{Conclusion}

Moving Target Defense research has matured around which configurations to deploy, how to switch, and when---while treating the candidate pool's own indistinguishability as an assumption rather than a property to verify. This paper made that property both verifiable and enforceable. We formalized pool safety as the largest $\varepsilon$-close subset of the configuration product, gave algorithms that construct that subset at every realistic scale, and measured the framework end-to-end on live cloud deployments. Both rotations we audited fall short of their nominal anonymity, and the shortfalls carry a general lesson: whether an axis contributes anonymity is decided by the deployment that hosts it, not by the axis itself. A defender should therefore audit each axis's latency gaps against the deployment's own noise floor, and rotate among the largest pool the algorithms certify indistinguishable rather than among every nominal variant. Two extensions remain open. The guarantee is per-observation: dwell-time limits and randomized switching already bound what a patient adversary gains by averaging (\S\ref{sec:cloud-setup}), but the two compose with the $\varepsilon$-close pool by argument rather than by proof, and a distributional guarantee over entire rotation sequences would make that composition rigorous. The guarantee is likewise per-channel: an adversary denied the latency channel will pivot to response size or CPU footprint, and enforcing indistinguishability jointly across observable channels would close that pivot.

\bibliographystyle{plainnat}
\bibliography{references}

\newpage

\appendix

\section{Proofs}\label{sec:proofs}

\subsection{Proof of \autoref{prop:reduction} (Pairwise to Range Reduction)}

\begin{proof}
($\Rightarrow$) Suppose $|U(c_1) - U(c_2)| \leq \varepsilon$ for all $c_1, c_2 \in C$.
In particular, this holds for the pair achieving the maximum and minimum utilities in $C$, so
\[
  \max_{c \in C} U(c) - \min_{c \in C} U(c) \leq \varepsilon.
\]

($\Leftarrow$) Suppose $\max_{c \in C} U(c) - \min_{c \in C} U(c) \leq \varepsilon$.
For any $c_1, c_2 \in C$, we have
\[
  \min_{c \in C} U(c) \leq U(c_1), U(c_2) \leq \max_{c \in C} U(c),
\]
and therefore
\[
  |U(c_1) - U(c_2)| \leq \max_{c \in C} U(c) - \min_{c \in C} U(c) \leq \varepsilon. \qedhere
\]
\end{proof}

\subsection{Proof of \autoref{cor:budget} (Per-Axis Budget)}

\begin{proof}
For a product set $C = C_1 \times \cdots \times C_k$, the additive form~\eqref{eq:utility} lets the maximum and minimum of $U$ over $C$ be taken coordinate-wise:
\[
  \max_{c \in C} U(c) = \sum_{i=1}^{k} \max_{s \in C_i} w_i\, u(s),
  \qquad
  \min_{c \in C} U(c) = \sum_{i=1}^{k} \min_{s \in C_i} w_i\, u(s),
\]
since each coordinate of the maximizing (resp.\ minimizing) configuration can be chosen independently within its $C_i$. Subtracting the two identities gives
\[
  \max_{c \in C} U(c) - \min_{c \in C} U(c) = \sum_{i=1}^{k} \sigma_i(C_i),
\]
and applying \autoref{prop:reduction} to the left-hand side yields the claimed equivalence. \qedhere
\end{proof}

\subsection{Proof of \autoref{prop:sampling} (Uniform Window-Density Bound)}

\begin{proof}
We draw $M$ configurations $c^{(1)}, \ldots, c^{(M)}$ independently and uniformly at random from the Cartesian product $S_1 \times \cdots \times S_k$.
Let $F$ be the CDF of the random variable $U(c)$ and $F_M$ its empirical counterpart on the $M$ samples.
The Dvoretzky--Kiefer--Wolfowitz (DKW) inequality \citep{dvoretzky1956asymptotic}, with the tight constant of \citet{massart1990tight}, states that
\[
  \Pr\!\left[\sup_{x \in \mathbb{R}} |F_M(x) - F(x)| > t\right] \leq 2e^{-2Mt^2}.
\]
Crucially, this bound is \emph{uniform} over all $x \in \mathbb{R}$, so it applies simultaneously to every candidate window---including windows whose endpoints are chosen as a function of the samples.
Setting the right-hand side to $\alpha$ gives $t = \sqrt{\tfrac{1}{2M}\log\tfrac{2}{\alpha}}$, and with probability at least $1-\alpha$,
\begin{equation}\label{eq:dkw-uniform}
  \sup_{x \in \mathbb{R}} |F_M(x) - F(x)| \leq t.
\end{equation}

Because $U$ takes finitely many values, $F$ has atoms, and the closed window $[\ell, \ell + \varepsilon]$ must be written with a left limit: $p(\ell) = F(\ell + \varepsilon) - F(\ell^-)$ and $\hat p_M(\ell) = F_M(\ell + \varepsilon) - F_M(\ell^-)$, where $F(\ell^-) = \lim_{y \uparrow \ell} F(y)$. The uniform bound extends to left limits, since $|F_M(x^-) - F(x^-)| = \lim_{y \uparrow x} |F_M(y) - F(y)| \leq \sup_{y} |F_M(y) - F(y)|$. By the triangle inequality, for any window $[\ell, \ell + \varepsilon]$,
\[
  |\hat p_M(\ell) - p(\ell)|
  \leq |F_M(\ell + \varepsilon) - F(\ell + \varepsilon)| + |F_M(\ell^-) - F(\ell^-)|
  \leq 2 \sup_{x} |F_M(x) - F(x)|
  \leq 2t.
\]
Taking the supremum over $\ell$ on the left-hand side preserves the bound (the right-hand side does not depend on $\ell$), giving the proposition:
\[
  \sup_{\ell \in \mathbb{R}} |\hat p_M(\ell) - p(\ell)| \leq 2t \quad \text{with probability } \geq 1-\alpha.
\]
Solving $2t \leq \varepsilon_{\text{stat}}$ for $M$ yields $M \geq \tfrac{2}{\varepsilon_{\text{stat}}^{2}} \log(2/\alpha) = O(\varepsilon_{\text{stat}}^{-2}\log(1/\alpha))$.
\end{proof}

\subsection{Proof of \autoref{cor:selection} (Selection-Bias-Aware Near-Optimality)}

\begin{proof}
Both argmaxes are attained: since $U$ takes finitely many values, the window contents---and hence $\hat p_M(\ell)$ and $p(\ell)$---take only finitely many values as $\ell$ ranges over $\mathbb{R}$.
The sliding-window pass returns such an argmax: shifting any window right until its left endpoint meets its smallest contained sample loses no samples, so some sample-anchored window---exactly the candidates the two-pointer pass enumerates---achieves the global maximum of $\hat p_M$.
Condition on the event in~\eqref{eq:dkw-uniform}, which holds with probability $\geq 1-\alpha$, so that $|\hat p_M(\ell) - p(\ell)| \leq 2t$ for every $\ell \in \mathbb{R}$ simultaneously.
By the definition of $\hat\ell$ as the empirical argmax,
\begin{equation}\label{eq:argmax}
  \hat p_M(\hat\ell) \;\geq\; \hat p_M(\ell^*).
\end{equation}
Apply uniform convergence at $\ell^*$ and at $\hat\ell$:
\begin{align*}
  \hat p_M(\ell^*) &\geq p(\ell^*) - 2t \;=\; p^* - 2t, \\
  p(\hat\ell)      &\geq \hat p_M(\hat\ell) - 2t.
\end{align*}
Chaining these with~\eqref{eq:argmax} gives
\[
  p(\hat\ell) \;\geq\; \hat p_M(\hat\ell) - 2t
              \;\geq\; \hat p_M(\ell^*) - 2t
              \;\geq\; p^* - 4t,
\]
i.e., $p^* - p(\hat\ell) \leq 4t$.
Substituting $t = \sqrt{\tfrac{1}{2M}\log\tfrac{2}{\alpha}}$ yields the stated bound.
The selection bias is absorbed entirely by the supremum in~\eqref{eq:dkw-uniform}: because the DKW bound is uniform, it covers \emph{any} window the algorithm might choose from the samples, including the empirical argmax.
\end{proof}

\section{Empirical validation of the selection-bias bound}\label{sec:selbias-app}

\autoref{cor:selection} guarantees that, despite reusing samples to both select the densest window and estimate its density, the \emph{true} density of the empirically-selected window is within $4t = 4\sqrt{\tfrac{1}{2M}\log\tfrac{2}{\alpha}}$ of the optimum with probability $\geq 1-\alpha$.
We verify this empirically on \textsc{Medium} (where $p^*$ is computable exactly).
For each $M \in \{10^3, 3{\cdot}10^3, 10^4, 3{\cdot}10^4, 10^5, 3{\cdot}10^5\}$ we run $R{=}200$ replicates of the sampling algorithm and, for each replicate, compute the gap $p^* - p(\hat\ell)$, where $p(\hat\ell)$ is the true count of configurations falling in the empirically-selected window divided by $N$.

\begin{figure}[t]
\centering
\includegraphics[width=\linewidth]{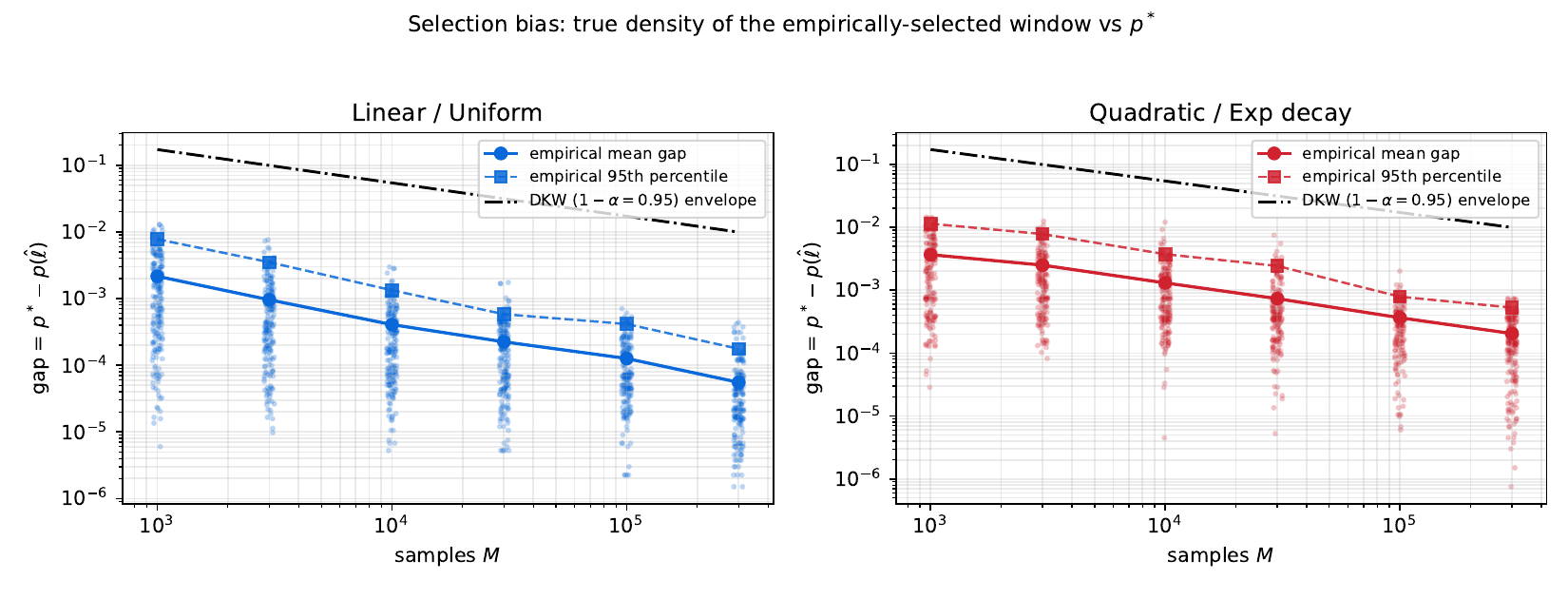}
\caption{Empirical selection-bias gap $p^* - p(\hat\ell)$ as a function of sample count $M$ on \textsc{Medium} (200 replicates per $M$).
Dots: individual replicates (jittered).
Solid line: empirical mean.
Dashed line: empirical $95$th percentile.
Dot-dashed: the DKW envelope from \autoref{cor:selection} at $\alpha{=}0.05$.
The empirical gaps shrink at the predicted $M^{-1/2}$ rate; no replicate violates the envelope, which is conservative by $1$--$2$ orders of magnitude.}
\label{fig:selbias}
\end{figure}

\autoref{fig:selbias} supports three conclusions.
First, the DKW envelope is never violated by any of the $200$ replicates at any $M$, so the worst-case bound holds in practice.
Second, the empirical mean and 95th-percentile gaps both decay at approximately $M^{-1/2}$ (slope $-0.5$ on log--log axes), matching the theoretical scaling.
Third, the constant in the DKW bound is conservative by roughly one to two orders of magnitude: at $M{=}10^5$ the empirical 95th percentile is $\sim 4\times 10^{-4}$ versus a DKW envelope of $\sim 1.7\times 10^{-2}$, so practitioners running with even modest $M$ obtain near-optimal pool sizes far more reliably than the worst-case bound implies.

\section{Algorithm details}\label{sec:algorithms-app}

This appendix expands each algorithm summarized in \autoref{tab:complexity} with pseudocode, complexity derivation, and---where relevant---error analysis. All four algorithms share the same final step: a two-pointer sliding window over a sorted array of utility values to find the densest $\varepsilon$-width interval.

\subsection{Full enumeration with iterative outer sums}
\label{sec:app-enum}

Materializing all $N = \prod_i n_i$ configuration tuples costs $O(kN)$ space and $O(kN)$ evaluation time, a factor-$k$ overhead over handling the utility values alone. We instead keep only the running sumset of utilities. Starting from $V_1 = \{w_1 \cdot u(s) : s \in S_1\}$, we form
\[
    V_{1:i+1} = \{a + b : a \in V_{1:i},\, b \in V_{i+1}\}
    \qquad \text{for } i = 1, \ldots, k-1,
\]
so after $k-1$ iterations $V_{1:k}$ is the full multiset $\{U(c) : c \in S_1 \times \cdots \times S_k\}$. No configuration tuple is materialized, and storage never exceeds $O(N)$ scalars. Sorting $V_{1:k}$ and applying a two-pointer sliding window of width $\varepsilon$ finds the densest interval in a single linear pass, so sorting dominates: $O(N \log N)$ time, $O(N)$ space.

\begin{tcolorbox}[
  enhanced, breakable,
  colback=bgsurface,
  colframe=bordercolor,
  title={\texttt{\textcolor{textmuted}{Algorithm 1: Full Enumeration + Sliding Window}}},
  fonttitle=\small,
  colbacktitle=bgsurface2,
  coltitle=textmuted,
  arc=2pt, top=0pt, bottom=0pt
]
\begin{lstlisting}[style=darkcode, language=pseudocode, mathescape=true]
Input:  Sets $S_1$, ..., $S_k$; weights $w_1$, ..., $w_k$;
        tolerance $\varepsilon$
Output: Size $|C|$ and interval bounds $[l, l + \varepsilon]$

// Step 1: Compute per-set weighted utilities
for $i = 1$ to $k$:
    $V_i$ <- {$w_i \cdot u(s)$ : $s \in S_i$}

// Step 2: Initialize accumulator
$A$ <- $V_1$

// Step 3: Iteratively compute outer sums
for $i = 2$ to $k$:
    $A$ <- {$a + b$ : $a \in A$, $b \in V_i$}

// Step 4: Sort by utility value
Sort $A$ in non-decreasing order

// Step 5: Apply sliding window
Apply two-pointer window to find densest
$\varepsilon$-width interval

// Step 6: Return result
Return size $|C|$ and bounds $[l, l + \varepsilon]$
\end{lstlisting}
\end{tcolorbox}

\subsection{Meet-in-the-middle}
\label{sec:app-mitm}

When $N$ exceeds direct enumeration, we partition the $k$ sets into two halves $A_{\text{sets}} = (S_1, \ldots, S_{\lfloor k/2 \rfloor})$ and $B_{\text{sets}} = (S_{\lfloor k/2 \rfloor + 1}, \ldots, S_k)$. Enumeration proceeds independently on each half using the outer-sum trick of \S\ref{sec:app-enum}, producing sumsets $V_A$ and $V_B$ of sizes $|V_A|, |V_B| \approx \sqrt{N}$ when the components are balanced.

The full utility multiset $\{U(c) : c \in S_1 \times \cdots \times S_k\}$ equals the pairwise sum $\{a + b : a \in V_A,\, b \in V_B\}$: each configuration $c$ maps to a unique $(a, b)$ pair with $U(c) = a + b$, and each pair contributes exactly one configuration. Forming these sums, sorting the result, and applying the sliding window returns the exact optimum. The memory advantage is that neither half's enumeration ever exceeds $O(\sqrt{N})$; the final $O(N)$ arises only in the merge. If even the merge does not fit in memory, the pairwise sums can be generated in sorted order from the two sorted halves with a heap of size $O(\sqrt{N})$, so the sliding window can be streamed while holding only the values inside the current $\varepsilon$-width window. This is the classical technique of \citet{horowitz1974computing}, which halves the exponent of the enumeration cost, adapted here to the range-density query.

\begin{tcolorbox}[
  enhanced, breakable,
  colback=bgsurface,
  colframe=bordercolor,
  title={\texttt{\textcolor{textmuted}{Algorithm 2: Meet in the Middle}}},
  fonttitle=\small,
  colbacktitle=bgsurface2,
  coltitle=textmuted,
  arc=2pt, top=0pt, bottom=0pt
]
\begin{lstlisting}[style=darkcode, language=pseudocode, mathescape=true]
Input:  Sets $S_1$, ..., $S_k$; weights $w_1$, ..., $w_k$;
        tolerance $\varepsilon$
Output: Size $|C|$ and interval bounds $[l, l + \varepsilon]$

// Step 1: Partition into two halves
$m$ <- $\lfloor k/2 \rfloor$

// Step 2: Enumerate each half via outer sums
$V_A$ <- outer-sum enumeration over $(S_1, ..., S_m)$
$V_B$ <- outer-sum enumeration over $(S_{m+1}, ..., S_k)$

// Step 3: Merge halves via pairwise sums
$A$ <- {$a + b$ : $a \in V_A$, $b \in V_B$}

// Step 4: Sort merged multiset
Sort $A$ in non-decreasing order

// Step 5: Apply sliding window of width $\varepsilon$
Find densest $\varepsilon$-interval $[l, l + \varepsilon]$

// Step 6: Return result
Return size $|C|$ and bounds $[l, l + \varepsilon]$
\end{lstlisting}
\end{tcolorbox}

\subsection{FFT convolution: correspondence and discretization}
\label{sec:app-fft}

\paragraph{Convolution correspondence.} Discretize per-component utilities into $R$ bins of width $\delta$. Let $h_i(v) = |\{s \in S_i : \text{bin}(w_i \cdot u(s)) = v\}|$ be the histogram of set $S_i$ after discretization. The number of configurations whose discretized total utility falls in bin $z$ is
\[
    \#\{c : \text{bin}(U(c)) = z\}
    = \sum_{v_1 + \cdots + v_k = z} \prod_{i=1}^{k} h_i(v_i)
    = (h_1 * h_2 * \cdots * h_k)(z),
\]
where $*$ denotes discrete convolution. The product form arises because each set contributes independently: fixing the per-component bins $(v_1, \ldots, v_k)$ leaves $\prod_i h_i(v_i)$ independent configuration choices \citep{cormen2009introduction}. Computing $k - 1$ pairwise convolutions via FFT costs $O(k R \log R)$ \citep{cooley1965algorithm}. A prefix-sum sliding window over the resulting histogram then identifies the densest $\varepsilon$-interval in $O(R)$ additional work.

\paragraph{Discretization error.} Rounding each per-component contribution to the nearest bin displaces it by at most $\delta/2$, so the total utility of any configuration is shifted by at most $k\delta/2$. This shift can move a configuration across the boundary of the optimal $\varepsilon$-window only if its true utility lies within $k\delta/2$ of the boundary, and the mass within $k\delta/2$ of the two boundaries is an $O(k\delta/\varepsilon)$ fraction of the window's own mass whenever the utility distribution has bounded density near the boundaries. Choosing $\delta = O(\varepsilon/k)$ therefore guarantees a fixed relative accuracy at the cost of $R = O((u_{\max} - u_{\min}) \cdot k / \varepsilon)$ bins. In practice $\delta = 0.05$ keeps relative count error under $5\%$ for all datasets in \autoref{tab:datasets}.

\begin{tcolorbox}[
  enhanced, breakable,
  colback=bgsurface,
  colframe=bordercolor,
  title={\texttt{\textcolor{textmuted}{Algorithm 3: FFT Convolution + Window Search}}},
  fonttitle=\small,
  colbacktitle=bgsurface2,
  coltitle=textmuted,
  arc=2pt, top=0pt, bottom=0pt
]
\begin{lstlisting}[style=darkcode, language=pseudocode, mathescape=true]
Input:  Sets $S_1$, ..., $S_k$; weights $w_1$, ..., $w_k$;
        tolerance $\varepsilon$; resolution $\delta$
Output: Count estimate and interval bounds $[l, l + \varepsilon]$

// Step 1: Per-set utilities, range, and bin count
for $i = 1$ to $k$:
    $V_i$ <- {$w_i \cdot u(s)$ : $s \in S_i$}
$u_{\min}$ <- $\sum_i \min V_i$;  $u_{\max}$ <- $\sum_i \max V_i$
$R$ <- ceil$((u_{\max} - u_{\min}) / \delta)$ + 1

// Step 2: Build per-set utility histograms
for $i = 1$ to $k$:
    $h_i$ <- histogram of $V_i$ on $R$ bins

// Step 3: Initialize convolution
$H$ <- $h_1$

// Step 4: Iteratively convolve histograms
for $i = 2$ to $k$:
    $H$ <- FFTConvolve($H$, $h_i$)  // zero-padded FFT

// Step 5: Apply sliding window to result
Compute prefix sums of $H$
Sweep window of width ceil$(\varepsilon / \delta)$

// Step 6: Return estimate
Return count estimate and bounds $[l, l + \varepsilon]$
\end{lstlisting}
\end{tcolorbox}

\subsection{Monte Carlo sampling procedure}
\label{sec:app-sampling}

When $R$ is prohibitively large (e.g., $k \geq 15$ with wide per-component ranges), even the FFT becomes impractical. The procedure is: draw $M$ configurations i.i.d.\ uniformly from $S_1 \times \cdots \times S_k$ (equivalently, sample one $s_i$ from each $S_i$ independently), compute $U(c)$ for each sample, sort the utilities, and apply the sliding window to identify the empirically densest interval $[\hat\ell, \hat\ell + \varepsilon]$. Scaling the sample count in $[\hat\ell, \hat\ell + \varepsilon]$ by $N/M$ estimates the true count.

Random configuration sampling in additive-utility spaces has precedent in software product lines \citep{oh2017near}. Because $U$ is a sum of $k$ independent per-component contributions, its standard deviation grows as $\sqrt{k}$ while its range grows linearly in $k$, so the distribution concentrates around its mean as $k$ grows and samples land densely in the modal region where the densest window typically lies. The uniform-density and selection-bias guarantees for this algorithm are Prop.~\ref{prop:sampling} and Cor.~\ref{cor:selection} in \S\ref{sec:approach}; proofs are in \S\ref{sec:proofs}.

\begin{tcolorbox}[
  enhanced, breakable,
  colback=bgsurface,
  colframe=bordercolor,
  title={\texttt{\textcolor{textmuted}{Algorithm 4: Monte Carlo Sampling}}},
  fonttitle=\small,
  colbacktitle=bgsurface2,
  coltitle=textmuted,
  arc=2pt, top=0pt, bottom=0pt
]
\begin{lstlisting}[style=darkcode, language=pseudocode, mathescape=true]
Input:  Sets $S_1$, ..., $S_k$; weights $w_1$, ..., $w_k$;
        tolerance $\varepsilon$; sample count $M$
Output: Estimated count and interval bounds $[l, l + \varepsilon]$

// Step 1: Draw M i.i.d. uniform samples
for $j = 1$ to $M$:
    for $i = 1$ to $k$:
        sample $s_i$ ~ Uniform($S_i$)
    $u^{(j)}$ <- $\sum_{i=1}^{k} w_i \cdot u(s_i)$

// Step 2: Sort sample utilities
Sort $(u^{(1)}, ..., u^{(M)})$ non-decreasing

// Step 3: Sliding window over sorted samples
Find $\hat l$ maximizing sample count
in $[\hat l, \hat l + \varepsilon]$

// Step 4: Scale to estimate true count
$\hat{|C|}$ <- (count in window) $\times$ $N / M$

// Step 5: Return estimate
Return $\hat{|C|}$ and bounds $[\hat l, \hat l + \varepsilon]$
\end{lstlisting}
\end{tcolorbox}

\end{document}